\documentclass[% reprint,
nofootinbib,
 amsmath,amssymb,
 aps,
pra,
onecolumn
]{revtex4-2}

\usepackage{graphicx}
\usepackage{amssymb,amsfonts,amsmath,amsthm}
\usepackage{braket}
\usepackage{hyperref}
\usepackage{color}
\usepackage{soul}
\usepackage{changes}

\usepackage{enumerate,enumitem}
\usepackage{comment}
\usepackage{hyperref}
\usepackage{eqnarray}
\usepackage{bbold}
\usepackage{optidef}
\usepackage{booktabs}

\newtheorem{definition}{Definition}
\newtheorem{theorem}{Theorem}

\newtheorem{lemma}{Lemma}

\def\D{\mathcal{D}}
\def\F{\mathcal{F}}
\def\H{\mathcal{H}}
\def\O{\mathcal{O}}
\def\P{\mathcal{P}}
\def\S{\mathcal{S}}
\def\U{\mathcal{U}}
\def\N{\mathcal{N}}
\def\R{\mathbb{R}}

\def\rhoab{\rho^{AB}}
\def\sigmaab{\sigma^{AB}}
\def\Tr{\mathrm{Tr}}

\def\roof{\mathrm{cr}}
\def\top{\mathrm{top}}
\def\Ud{\U^{\D}}
\def\Uo{\U^{\O}}

\def\sn{\mathcal{N}}
\def\rk{\mathcal{R}}
\def\psetord{\Delta^\downarrow_{d}} 
\def\Ha{\H^A}
\def\Hb{\H^B}
\def\Hab{\H^{AB}}

\newcommand{\locc}[0]{\underset{\mathrm{LOCC}}{\rightarrow}}
\DeclareMathOperator{\supp}{supp}
\DeclareMathOperator{\tr}{tr}
\DeclareMathOperator{\rank}{rank}
\DeclareMathOperator{\eig}{eig}

\newcommand{\rvline}{\hspace*{-\arraycolsep}\vline\hspace*{-\arraycolsep}}

\begin{document}

\title{Extending Schmidt vector from pure to mixed states for characterizing entanglement}

\author{F. Meroi}
\affiliation{Facultad de Ciencias Exactas, Ingenier\'ia y Agrimensura, Universidad Nacional de Rosario, Rosario, Argentina}

\author{M. Losada}
\affiliation{Facultad de Matemática, Astronomía, Física y Computación , Universidad Nacional de C\'{o}rdoba, C\'ordoba, Argentina}

\author{G.M. Bosyk}
\email{gbosyk@gmail.com}
\affiliation{CONICET-Universidad de Buenos Aires, Instituto de Ciencias de la Computación (ICC), Buenos Aires, Argentina}

\date{\today}

\begin{abstract}

In this study, we enhance the understanding of entanglement transformations and their quantification by extending the concept of Schmidt vector from pure to mixed bipartite states, exploiting the lattice structure of majorization. 
The Schmidt vector of a bipartite mixed state is defined using two distinct methods: as a concave roof extension of Schmidt vectors of pure states, or equivalently, from the set of pure states that can be transformed into the mixed state through local operations and classical communication (LOCC). 
We demonstrate that the Schmidt vector fully characterizes separable and maximally entangled states.  
Furthermore, we prove that the Schmidt vector is monotonic and strongly monotonic under LOCC, giving necessary conditions for conversions between mixed states.
Additionally, we extend the definition of the Schmidt rank from pure states to mixed states as the cardinality of the support of the Schmidt vector and show that it is equal to the Schmidt number introduced in previous work [Phys. Rev. A \textbf{61}, 040301 (R), 2000]. 
Finally, we introduce a family of entanglement monotones by considering concave and symmetric functions applied to the Schmidt vector.

\end{abstract}

\maketitle

\section{Introduction}

Entanglement is a fundamental concept in quantum theory~\cite{Horodecki2009,Guhne2009} and a relevant resource for diverse quantum information applications, such as quantum teleportation~\cite{Bennett1993}, quantum key distribution~\cite{Ekert1991,Curty2004}, quantum networks\cite{Azuma2023}, among others. 
Its detection, quantification, and characterization are crucial not only for theoretical reasons but also for practical applications. 

In the realm of bipartite pure states, these issues are fully understood using the Schmidt decomposition of the state~\cite{NielsenBook}. 
In this case, the Schmidt vector, defined as the vector formed by the squares of the Schmidt coefficients, and the Schmidt rank, representing the size of its support, play key roles in characterizing entanglement. 
Specifically, a Schmidt rank greater than one indicates the presence of entanglement and vice versa. 
Furthermore, converting an initial state to a final state through local operations and classical communication (LOCC) is completely characterized by a majorization relation between the Schmidt vectors of the states~\cite{Nielsen1999}. 
Remarkably, any measure of entanglement for bipartite pure states can be expressed as a symmetric and concave function applied to the Schmidt vector~\cite{Vidal2000}.
Additionally, the largest Schmidt coefficient of a given bipartite pure state is useful for constructing fidelity-based witness operators~\cite{Bourennane2004}.

However, in realistic scenarios, states are affected by noise, and due to phenomena like decoherence, the state is generally mixed. This makes entanglement characterization more complex. A natural question arises: can we define a notion of Schmidt vector that is useful for characterizing entanglement even for mixed states? 

We provide a positive answer to this problem by exploiting the majorization lattice techniques developed in \cite{Bosyk2021} for the case of quantum coherence. Specifically, we extend the notion of a Schmidt vector from pure to mixed states using two distinct methods that yield the same vector: 
(i) as a concave roof extension of Schmidt vectors for pure states, or (ii) as the supremum in the majorization lattice of all Schmidt vectors for pure states that can be transformed into the given state via LOCC. 

We demonstrate that this Schmidt vector fully characterizes separable and maximally entangled states. 
Additionally, we prove that the Schmidt vector exhibits monotonic and strong monotonic properties under LOCC, providing necessary conditions for conversions between mixed states. 
Moreover, we propose the cardinality of the support of the Schmidt vector as an extension of the Schmidt rank from pure-to-mixed state, and we show that it is equivalent to the Schmidt number \cite{Terhal2000,Sanpera2001}.
Finally, we introduce a family of entanglement monotones by applying concave and symmetric functions to the Schmidt vector.

The paper is organized as follows. 
In Section~\ref{sec:prelimaries}, we review the main concepts of majorization (Sec.~\ref{sec:majorization}) and entanglement (Sec.\ref{sec:entanglement}) theories necessary for our purposes.
In Section~\ref{sec:results}, we provide our main results. 
First, we introduce the notion of the Schmidt vector of an arbitrary bipartite state and discuss its main properties (Sec.~\ref{sec:schimdt_vector}). 
Then, we show the equivalence between the Schmidt number and the cardinality of the support of the Schmidt vector (Sec.~\ref{sec:equivalence_Schmidtrank_Schmidtnumber}). We give an example of the Schmidt vector for a family of quantum states (Sec.~\ref{sec:example}) and we show how to construct a family of entanglement monotones from it (Sec.~\ref{sec:Schmidt_vector_entanglement_monotone}).
Finally, we give some concluding remarks in Section~\ref{sec:conclusions}. 
For readability, auxiliary lemmas and proofs, and an explanation of numerical calculations are presented separately in Appendices~\ref{app:aux_proofs}  and~\ref{app:numerical}, respectively.

\section{Preliminaries}
\label{sec:prelimaries}

\subsection{Majorization}
\label{sec:majorization}

Here, we recall the definition of majorization and some related results necessary for our purposes. 
For a more comprehensive review of majorization theory, see, for instance, \cite{MarshallBook} and its applications on quantum information, see, for example, \cite{Nielsen2001,Bellomo2019}.

We will restrict our attention to majorization between probability vectors. 
Let us define the set of $d$-dimensional probability vectors as $\Delta_{d}= \big\{(u_1, \ldots, u_{d}) \in \mathbb{R}^d: u_i \geq 0,  \sum_{i= 1}^{d} u_i =1 \big\}.$
Additionally, we denote  $u^{\downarrow} = (u^\downarrow _1, \ldots, u^\downarrow_d)$ as the probability vector $u$ with entries ordered non-increasingly, and we collect these vectors in the subset $\Delta^\downarrow_{d} \subseteq \Delta_{d}$.

The majorization relation is defined as follows.
\begin{definition}[\textbf{Majorization relation}]
Given $u,v \in \Delta_{d}$, we say that $u$ is majorized by $v$ (denoted as $u \preceq v$) if, and only if, $\sum_{i=1}^{k} u^\downarrow_i \leq \sum_{i=1}^{k} v^\downarrow_i$, for all $k \in \{1, \ldots, d-1\}$.
\end{definition}

Another way to address majorization is through the notion of the Lorenz curve for a probability vector $u \in \Delta_d$, denoted as $L_u$.
Specifically, it is defined as the curve $L_u: [0,d] \to [0,1]$  given by the linear interpolation of the points $\{(j,s_{j}(u^\downarrow))\}_{j= 0}^d$, where $s_j(u)=\sum^j_{i=1} u_i$ with the convention $s_0 = 0$. 
 This curve is an increasing and concave function. Majorization $u \preceq v$ holds if and only if $L_u(x) \leq L_v(x)$ for all $x \in [0,d]$.

From an order-theoretic point of view, the majorization relation is a partial order on the set $\Delta^\downarrow_{d}$, meaning it is a binary relation that is reflexive, antisymmetric, and transitive. 
Thus, the tuple $\braket{\Delta^\downarrow_{d},\preceq}$ is a partially order set (POSET). 
Moreover, the supremum and infimum of any subset of $\U \subset \psetord$ exist, denoted as  $\bigvee \U$ and $\bigwedge \U$, respectively. 
Therefore, we have the following result~\cite{Bapat1991,Bosyk2019}. 
\begin{theorem}[\textbf{Majorization lattice}]
The tuple $\braket{\Delta^\downarrow_{d},\preceq}$ is a complete lattice.    
\end{theorem}
The algorithms for computing the supremum and infimum of any subset of the majorization lattice can be found in \cite{Cicalese2002,Bosyk2019,Massri2020}.
Here we recall the algorithm for obtaining the supremum of a set $\U \subseteq \psetord$, which will be useful in the next section.
For each $j$, we obtain  $S_j = \sup\{s_j(u) : u \in \U \}$.
Then, we calculate the upper envelope of the polygonal curve given by the linear interpolation of the set of points $\{(j,S_{j})\}_{j= 0}^d$, with the convention $S_0 = 0$.
The result is the Lorenz curve of $\bigvee \U$, $L_{\bigvee \U}$. 
Finally, we have $\bigvee \U = (L_{\bigvee \U}(1),L_{\bigvee \U}(2)-L_{\bigvee \U}(1),\ldots,L_{\bigvee \U}(d)-L_{\bigvee \U}(d-1))$.

The majorization relation is closely related to Schur-concave functions (see, for example, \cite[I.3]{MarshallBook}), which are functions that reverse the preorder relation. 
\begin{definition}[\textbf{Schur-concavity}]
A function $f: \Delta_d \to \mathbb{R}$ is Schur-concave if, for all $u, v \in \Delta_d$ such that $u \preceq v$, $f(u) \geq f(v)$.
\end{definition}
Moreover, suppose that the function $f$ also satisfies $f(u) > f(v)$ whenever $u$ is strictly majorized by $v$, denoted as $u \prec v$ (that is, when $u \preceq v$ and $u^\downarrow \neq v^\downarrow$). In that case, we say that it is strictly Schur-concave.
In particular, generalized entropies, including Shannon, Rényi, and Tsallis entropies, are examples of strictly Schur-concave functions (see, for example, \cite{Bosyk2016}).
Additionally, we have the following result that relates concavity and Schur-concavity~\cite[Prop.C.2.]{MarshallBook}.
\begin{theorem}[\textbf{Symmetry and concavity implies Schur-concavity}]
If $f$ is a symmetric and concave function, then $f$ is Schur-concave. 
\end{theorem}

In particular, if $f$ is symmetric and strictly concave, it is strictly Schur-concave.

\subsection{Entanglement}
\label{sec:entanglement}

Here, we will review the most significant results in entanglement theory. For a more detailed review, see for instance \cite{Horodecki2009,Guhne2009}. 

Consider a composite Hilbert space  $\Hab = \Ha \otimes \Hb$ with dimension $d^{AB} = d^{A}d^{B}$, where $d^A=\dim{\Ha}$ and $d^A=\dim{\Hb}$. 
Let us denote by $\S(\Hab)$ the set of all density matrices that act on $\Hab$, and by $\P(\Hab)$ the subset of pure states. 

We recall the definitions of separability and entanglement for bipartite quantum states.
\begin{definition}[\textbf{Separability and Entanglement}]

We say that $\rho^{AB} \in \S(\H^{AB})$ is separable if and only it can be expressed as a convex combination of product states, that is,
\begin{equation}
    \label{eq:sep_def}
    \rho^{AB}= \sum_i p_i \rho_i^A \otimes \rho_i^B,
\end{equation}
where $p_i \geq 0$, $\sum_i p_i=1$, and $\rho^A$ and $\rho^B$ are states that act on $\H^A$ and $\H^B$, respectively. 
If $\rho^{AB}$ is not separable, we call it entangled state.
\end{definition}

We now review the axiomatic approach for quantifying entanglement, which involves identifying the necessary and desirable properties that any entanglement measure should satisfy.
\begin{definition}[\textbf{Entanglement measure}]
\label{def:coherence_measures}
An entanglement measure is a function $E : \S(\H^{AB}) \to \R$ that satisfies the following conditions:
\begin{enumerate}
	 \item Non-negativity: $E\left(\rho^{AB} \right)  \geq 0$ and $E\left(\rho^{AB} \right)  = 0$ for separable states. \label{c1:nonneg}
	\item It is monotonic under LOCC:
    $E(\rhoab) \geq E(\Lambda(\rhoab))$ for all state $\rhoab \in \S(\H^{AB})$ and     all LOCC operation $\Lambda$.  \label{c2:monotonicity}
	\item Satisfies strong monotony under LOCC operations: $E(\rhoab) \geq \sum_{n = 1}^{N}  p_n E(\sigma^{AB}_n)$, for all state $\rhoab \in \S(\H^{AB})$ and    all LOCC operation $\Lambda$ given by Kraus operators $\{K_n\}_{n= 1}^N$, where $p_n = \Tr{K_n \rhoab K^\dag_n}$ and $\sigma^{AB}_n = K_n \rhoab K^\dag_n / p_n$.  \label{c3:strong_monotonicity}
	\item It reaches its maximum value in maximally entangled states: $\arg\max_{\rhoab \in S(\H^{AB})} E(\rhoab)$ is the set of pure states of the form $\rhoab_{\max} = \ket{\psi_{\max}}\bra{\psi_{\max}}$ with $\ket{\psi_{\max}} = \sum^d_{i=1} 1/\sqrt{d} \ket{a_i} \ket{b_i}$ and $d = \min\{\dim{\H^A},\dim{\H^B}\}$.  \label{c4:normalization}
	\item It is convex: $E\left( \sum_{i = 1}^{M}  q_i \rhoab_i\right)  \leq \sum_{i= 1}^{M} q_i E(\rhoab_i)$ .  \label{c5:convexity}
\end{enumerate}
\end{definition}
If $E$ satisfies conditions \ref{c1:nonneg} to \ref{c4:normalization}, we call it an entanglement monotone.

For bipartite pure states, the Schmidt decomposition provides the necessary tools to fully characterize the entanglement (see, for example, \cite[Th. 2.7]{NielsenBook}).

\begin{theorem}[\textbf{Schmidt decomposition}]
\label{th:Schmidt_deco}
%Let us consider two Hilbert spaces $\H^A$ and $\H^B$ of dimensions $d^A$ and $d^B$ respectively.
For any $\ket{\psi} \in \Hab= \Ha \otimes \Ha$ there are orthonormal sets $\{\ket{a_1},\ldots,\ket{a_{d^A}}\}$ of $\Ha$ and $\{\ket{b_1},\ldots,\ket{b_{d^B}}\}$ of $\Hb$, and coefficients (Schmidt coefficients) $\{\psi_i\}^d_{i=1}$ non-negative, unique (unless rearrangements) such that
\begin{equation}
    \label{eq:Schmidt_deco}
    \ket{\psi} = \sum^d_{i=1} \sqrt{\psi_i} \ket{a_i} \ket{b_i},
\end{equation}
with $\sum^d_{i=1} \psi_i=1$ and $d=\min\{d^A,d^B\}$.
\end{theorem}

We will see that the notion of the Schmidt vector plays a fundamental role in obtaining entanglement measures and characterizing entanglement transformations by LOCC.
For bipartite pure states, the Schmidt vector can be defined as follows~\cite{Zhu2017}.
\begin{definition}[\textbf{Schmidt vector}]
	\label{def:vector_Schmidt_pure}
    The Schmidt vector of a bipartite pure state $\ket{\psi}\bra{\psi}$ is defined as
	\begin{equation}
	\label{eq:vector_Schmidt_pure}
	\mu(\ket{\psi}\bra{\psi}) =  \left(\psi_1, \ldots, \psi_d\right),
	\end{equation}
        where $\{\psi_i\}^d_{i=1}$ are the Schmidt coefficients of $\ket{\psi}$.
\end{definition}

Alternatively, the Schmidt vector can be calculated as the eigenvalues of the reduced density matrix
	\begin{equation}
	\label{eq:vector_Schmidt_pure_PT}
	\mu(\ket{\psi}\bra{\psi}) =  \eig\left(\tr_{B}\left(\ket{\psi}\bra{\psi}\right)\right), 
	\end{equation}

where $\eig(\rho)$ denotes the vector formed by the eigenvalues of $\rho$ and $\tr_{B}$ the partial trace over subsystem $B$.  
 
The Schmidt rank, the number of non-null Schmidt coefficients, serves as a measure of entanglement dimension.
\begin{definition}[\textbf{Schmidt rank}]
	\label{def:Schmidt_rank_pure}
    The Schmidt rank of a pure bipartite state $\ket{\psi}\bra{\psi}$ is defined as
	\begin{equation}
	\label{eq:Schmidt_rank_pure}
        \rk\left(\ket{\psi}\bra{\psi} \right) = \rank\left(\tr_A \ket{\psi}\bra{\psi}\right).
        \end{equation}
        %
       % where $\tr_A$ denotes the partial trace over $\H_A$.
\end{definition}
The Schmidt rank has the following properties:
%
%\begin{observation}
%\hfil
    \begin{itemize}
        \item Equivalence between the Schmidt rank and the cardinality of the Schmidt vector support:
        \begin{equation}
        \label{eq:Schmidt_rank_pure_def2}
        \rk\left(\ket{\psi}\bra{\psi} \right) = \#\supp \mu(\ket{\psi}\bra{\psi}),
	\end{equation}
        where $\#A$ is the cardinality of set $A$ and $\supp(v) = \{i \in \mathbb{N}: v_i>0 \}$ is the support of vector $v$.    
        \item $\rk \left(\ket{\psi}\bra{\psi} \right)\leq d$.
        \item Entanglement characterization: $\ket{\psi}\bra{\psi}$ is entangled if and only if $\rk\left(\ket{\psi}\bra{\psi} \right)>1$.
        \item Monotonic under LOCC~\cite{Lo2001}: $\rk\left(\ket{\psi}\bra{\psi} \right) \geq\rk\left(\Lambda\left(\ket{\psi}\bra{\psi} \right)\right)$ for all state $\ket{\psi}\bra{\psi} \in \P(\H^{AB})$ and  all LOCC operation $\Lambda$.
    \end{itemize}
%\end{observation}
%

Symmetric and concave functions play also a fundamental role in quantifying bipartite entanglement. 
Let us first introduce the set of these functions.
\begin{definition}[\textbf{Symmetric and concave functions}]
We define $\F$ as the set of symmetric and concave functions that also satisfy the following conditions:
\begin{equation}
\label{eq:setofF}
\F= \big\{f : \mathbb{R}^d \to [0,1]: f \ \text{is symmetric and concave, } f(1, 0, \ldots, 0) =0 \ \text{and} \ \arg\max_{u \in \mathbb{R}^d} f(u)= (1/d,  \ldots, 1/d) \big\}.
\end{equation}
\end{definition}

An interesting result is that any entanglement measure restricted to pure states can be expressed in terms of a function from $\mathcal{F}$ applied to the Schmidt vector of the pure state \cite{Vidal2000}.
More specifically,
\begin{theorem}[\textbf{Pure case entanglement measures}]
\label{th:rnt_measures_pure}	
Given an entanglement measure $E : \S(\H^{AB})  \to \R$, there is a function $f_{E} \in \F$, 
such that the restriction of $E$ to pure states, denoted as $E|_{\P(\Hab)}$, can be written as
\begin{equation}
\label{eq:f_E}
E|_{\P(\Hab)}(\ket{\psi}\bra{\psi}) = f_{E}(\mu(\ket{\psi}\bra{\psi}).
\end{equation} 
\end{theorem}

On the other hand, in the opposite sense, it is possible to extend an entanglement measure defined for pure states to the general case. In the literature, at least two approaches have been proposed.
One of these methods was proposed in \cite{Vidal2000} (see, also, ~\cite{Zhu2017}), whereas the other was developed more recently in the papers \cite{Yu2021,Shi2021}.

The first proposal is based on the convex roof construction (see, for example, \cite{Uhlmann2010}).
Before introducing the definition of the ``convex roof-based entanglement measure'', we define $\D(\rho)$ as the set of all pure decompositions of a the state $\rho$:

\begin{definition}[\textbf{Pure state decompositions}]
Given $\rhoab$ an arbitrary bipartite state, we define the set
$\D(\rhoab)$ as
\begin{equation}
\label{eq:set_pure_ensambles}
\D(\rhoab)= \left\{ {\left\lbrace  q_i, \ket{ \psi_i} \right\rbrace }_{i = 1}^M : \ M \in \mathbb{N},  \ \rho= \sum_{i = 1}^{M}  q_i \ket{\psi_i}\bra{\psi_i}, \ \text{with} \ \ket{\psi_i} \in \H^{AB} \right\}, 	
\end{equation}
    
\end{definition}

Entanglement measures based on the convex roof technique are defined as follows, as presented in \cite{Zhu2017}, which are equivalent to the measures in \cite{Vidal2000}.
\begin{theorem}[\textbf{Convex roof entanglement measures}]
	\label{def:E_f_convex_roof}	
For all $f \in \F$,  $E^{\roof}_{f}: \S(\H^{AB})  \to \R$ given by 
	\begin{equation}
	\label{eq:convex_roof_mixed}
	E^{\roof}_{f}(\rhoab) = \inf_{\left\lbrace  q_i, \ket{ \psi_i} \right\rbrace_{i = 1}^M  \in \D(\rhoab) } \sum_{i=1}^M q_i 	 f(\mu(\ket{\psi_i}\bra{\psi_i}))
	\end{equation}
is an entanglement measure, that is, it satisfies the conditions~\ref{c1:nonneg}--\ref{c5:convexity}.
\end{theorem}

It is interesting to note that $E^{\roof}_{f}(\rhoab)$ is the largest of all entanglement measures, i.e.
\begin{equation}
    E^{\roof}_{f}(\rhoab) \geq E(\rhoab),
\end{equation}
for all $E$ entanglement measure.

As we mentioned before, the convex roof technique is not the only way to extend entanglement monotones from pure to mixed states. An alternative construction was recently proposed \cite{Yu2021,Shi2021}.
Before introducing this entanglement monotone, we need to define the set of all pure bipartite states that can be converted to a state $\rhoab$ by LOCC operations,

\begin{definition}[\textbf{Set of convertible pure  states}]
Given $\rhoab$ an arbitrary bipartite state, we define the set $\O(\rhoab)$ as
\begin{equation}
\label{eq:set_pure_to_rho}
\O(\rhoab)= \left\{ \ket{\psi}\bra{\psi} : \ket{\psi}\bra{\psi} \locc \rhoab \right\}.
\end{equation}
\end{definition}

We will call the following entanglement monotones \cite{Yu2021,Shi2021}, based on the previous set, top entanglement monotones.

\begin{theorem}[\textbf{Top entanglement monotone}]
	\label{def:E_f_convex_roof}	
For all $f \in \F$,  $E^{\top}_{f}: \S(\H^{AB})  \to \R$ given by
	\begin{equation}
	\label{eq:convertible_mixed}
	E^{\top}_{f}(\rhoab) =  \inf_{\ket{\psi}\bra{\psi} \in \O(\rhoab) } f(\mu(\ket{\psi}\bra{\psi}))
	\end{equation}
is an entanglement monotone, i.e. it satisfies the conditions~\ref{c1:nonneg}--\ref{c4:normalization}.
\end{theorem}

It is important to note that the function $E^{\top}_{f}(\rhoab)$ in general is not convex.
The expression top in its name comes from the fact that it is the largest among all  entanglement monotone, which can be expressed as follows \cite{Yu2021}:
\begin{equation}
    E^{\top}_{f}(\rhoab) \geq E(\rhoab),
\end{equation}
for all entanglement monotone $E$.

We now review some results on entanglement transformations that will be useful for the following section and are based on a majorization relation between Schmidt vectors. 
Let us begin with Nielsen theorem  \cite{Nielsen1999}, which characterizes pure-to-pure state LOCC conversions.

\begin{theorem}[\textbf{Pure-to-pure states LOCC conversions}]
	\label{th:nielsen_th}	
	Let $\ket{\psi}\bra{\psi}$ and $\ket{\phi}\bra{\phi}$ be two arbitrary bipartite pure states.
        Then,
	\begin{equation}
	\label{eq:nielsen_th}
	\ket{\psi}\bra{\psi} \locc \ket{\phi}\bra{\phi} \iff \mu(\ket{\psi}\bra{\psi}) \preceq \mu(\ket{\phi}\bra{\phi}).
	\end{equation}
	%
 %with $\preceq$ the majorization relation.
\end{theorem}

This result has been generalized to cases where the target state is an ensemble of pure states~\cite{Jonathan1999} and more recently to cases where the final state is mixed~\cite{Zanoni2024}.
\begin{theorem}[\textbf{Pure-to-mixed states LOCC conversions}]
	\label{th:plenio_th}	
     Let $\ket{\psi}\bra{\psi}$ be an arbitrary bipartite pure state and $\sigma^{AB}$ be an arbitrary bipartite state.
     Then,
	\begin{equation}
	\label{eq:plenio_th}
        \begin{split}
	\ket{\psi}\bra{\psi} \locc \sigma^{AB}  &\iff \exists \{  p_n, \ket{\phi_{n}}\}_{n= 1}^N \ \ \text{such that} \\
	   &\begin{array}{l}
		(1) \ \sigma^{AB}= \sum_{n = 1}^{N} p_n \ket{\phi_{n}}\bra{\phi_{n}} \ \text{and} \\
		(2) \ \mu(\ket{\psi}\bra{\psi})  \preceq \sum_{n = 1}^{N} p_{n} \mu^\downarrow(\ket{\phi_{n}}\bra{\phi_{n}}).
	   \end{array}
	\end{split}
	\end{equation}
	%
 %with $\preceq$ the majorization relation.
 \label{Jonathan and Pleino}
\end{theorem}

Finally, we recall that the Schmidt rank can be generalized to mixed states using the convex roof technique, resulting in what is known as the Schmidt number~\cite{Terhal2000}.
\begin{definition}[\textbf{Schmidt number}]
A bipartite density matrix $\rhoab \in \S(\H^{AB})$ has Schmidt number $\N(\rhoab)=n$ if 
 \begin{enumerate}
      \item For any pure state decomposition $\{q_i, \ket{\psi_i}  \}_{i = 1}^M$ ($M \in \mathbb{N}$) of $\rhoab$, there exists a pure state $\ket{\psi_{i^*}}$ such that $\rk \left(\ket{\psi_{i^*}}\bra{\psi_{i^*}} \right) \geq n$.

     \item There exists a pure state decomposition $\{q_i, \ket{\psi_i}\}_{i = 1}^M$ ($M \in \mathbb{N}$) of $\rhoab$ 
     such that $\rk \left(\ket{\psi_{i}}\bra{\psi_{i}} \right) \leq n$ for all $\ket{\psi_{i}}$ of the pure state decomposition. 
\end{enumerate}
\end{definition}
Separable states have Schmidt number equal to $1$ whereas entangled states have larger than $1$. Furthermore, the Schmidt number is a non-additive quantity that cannot increase under LOCC and it is invariant under local operations. 

Alternatively, the Schmidt number can be expressed as~\cite{Sanpera2001} 
\begin{equation}
    \N(\rhoab) = \min_{\left\lbrace  q_i, \ket{ \psi_i} \right\rbrace_{i = 1}^M  \in \D(\rhoab) } \max_i \rk(\ket{ \psi_i}\bra{ \psi_i}).
\end{equation}

\section{Results}
\label{sec:results}

\subsection{Schmidt vector of a density matrix: Definition and properties}
 \label{sec:schimdt_vector}

%The proposal aims to explore the feasibility of extending the concept of the Schmidt vector from pure states to mixed states,

We propose an extension of the Schmidt vector from pure states to mixed states, applying the techniques developed in our previous work \cite{Bosyk2021} on quantum coherence. 
Two extensions are considered.

The first extension is inspired by the convex roof method~\cite{Uhlmann2010}.
We assign a probability vector to 
each pure state decomposition of $\rhoab$, and we consider the set of all such vectors, $\Ud(\rhoab)$.

%For every pure state decomposition of $\rhoab$, we assign it a probability vector and collect all such vectors into the ensuing set $\Ud(\rhoab)$.

\begin{definition}[\textbf{Set $\Ud$}]
For any bipartite state $\rhoab$, we define the set
 \begin{equation}
     \Ud(\rhoab) = \left\{ \sum_{i = 1}^{M}  q_i \mu^\downarrow(\ket{\psi_i}\bra{\psi_i}) : \ M \in \mathbb{N}, \  \{q_i, \ket{\psi_i}\}_{i= 1}^M  \in \D(\rhoab) \right\}.
 \end{equation}
 
\end{definition}

Since $\Ud(\rhoab) \subseteq \psetord$, taking the supremum allows us to derive a probability vector associated with the state $\rhoab$, irrespective of its specific pure state decomposition.

\begin{definition}[\textbf{Schmidt vector: first extension}]
\label{def:Schmidt vector_D}
For any bipartite state $\rhoab$, we define the vector $\nu_{\D}(\rhoab)$ as follows
\begin{equation}
    \nu_{\D}(\rhoab) =  \bigvee \Ud(\rhoab), 
\end{equation}
where $\bigvee$ is the supremum given by the majorization relation.  
\end{definition}

The second extension is motivated by the definition of a top monotone \cite{Yu2021,Shi2021}. 
For each pure state that can be transformed into $\rhoab$ via LOCC, we assign a Schmidt vector and gather all such vectors into the resulting set $\Uo(\rhoab)$.

\begin{definition}[\textbf{Set $\Uo$}]
For any bipartite state $\rhoab$, we define the set
 \begin{equation}
     \Uo(\rhoab) =\left\{ \mu^\downarrow(\ket{\psi}\bra{\psi}) : \ket{\psi}\bra{\psi} \in \O(\rhoab)  \right\}.
 \end{equation}
 
\end{definition}

Similarly, as $\Uo(\rhoab)$ lies within $\psetord$, determining the supremum allows us to derive a probability vector associated with the state $\rhoab$, independently of the pure states of the set $\O(\rhoab)$.

\begin{definition}[\textbf{Schmidt vector: second extension}]
\label{def:Schmidt vector_O}
For any bipartite state $\rhoab$, we define the vector $\nu_{O}(\rhoab)$ as follows
\begin{equation}
    \nu_{\O}(\rhoab) = \bigvee \Uo(\rhoab), 
\end{equation}
where $\bigvee$ is the supremum given by the majorization relation. 

\end{definition}
In what follows, we show that $\nu_{\O}$ and $\nu_{\D}$ are equal. 
\begin{theorem}[\textbf{Equivalence of both extensions}]
\label{th:equivalence_extensions}
    For any bipartite state $\rhoab$ we have:
    \begin{equation}
    \label{Equiv}
           \nu_{\D}(\rhoab) =   \nu_{\O}(\rhoab).
    \end{equation}
\end{theorem}
%%%%

Since Definitions \ref{def:Schmidt vector_D} and \ref{def:Schmidt vector_O} are equivalent, we
define the Schmidt vector of $\rhoab$ as:

\begin{definition}
For a bipartite pure state $\rhoab$, we define its Schmidt vector $\nu (\rhoab)$ as
\begin{equation}
    \nu (\rhoab) = \nu_{\D} (\rhoab) = \nu_{\O} (\rhoab). 
\end{equation}
    
\end{definition}

The Schmidt vector $\nu(\rho^{AB})$ generalizes the Schmidt vector for mixed states. 
Specifically, if $\rho^{AB}$ is a pure state \(\ket{\psi}\bra{\psi}\), then $\mu^\downarrow(\ket{\psi}\bra{\psi})=\bigvee \Uo(\rhoab)$. 
Furthermore, it possesses relevant properties that reflect key aspects of the notion of entanglement:

\begin{theorem}[\textbf{Schmidt vector properties}]
\label{th:Schmidt_properties}
The Schmidt vector $\nu(\rhoab)$ has the following properties:
\begin{enumerate}
	\item Top only attainable for separable states: 
	$\rhoab$ is separable if and only if $\nu(\rhoab) = (1,0,\ldots, 0)$. \label{c1:sep_sch}\\
        \item Monotonicity under LOCC: $\nu(\rhoab) \preceq  \nu(\Lambda(\rhoab))$ where $\Lambda$ is a LOCC and $\rhoab$ any bipartite state.  \label{c2:mono_sch}\\
	\item Strong monotonicity under LOCC: $\nu(\rhoab) \preceq \sum_{n = 1}^{N}  p_n  \nu(\sigma^{AB}_n)$,  for all state $\rhoab$ and    all LOCC $\Lambda$ given by Kraus operators $\{K_n\}_{n= 1}^N$, where $p_n = \Tr{K_n \rhoab K^\dag_n}$ and $\sigma^{AB}_n = K_n \rhoab K^\dag_n / p_n$. \label{c3:strong_mono_sch}
        \item Bottom only attainable for maximally entangled states:  $\nu(\rhoab) = (1/d,\ldots, 1/d)$ if and only if   $\rhoab$ is of the form $\rhoab_{\max} = \ket{\psi_{\max}}\bra{\psi_{\max}}$ with $\ket{\psi_{\max}} = \sum^d_{i=1} 1/\sqrt{d} \ket{a_i} \ket{b_i}$.  \label{c4:maxent_sch}

\end{enumerate}

\end{theorem}

\subsection{Equivalence between Schmidt rank and Schmidt number}
\label{sec:equivalence_Schmidtrank_Schmidtnumber}

We can naturally extend the definition of the Schmidt rank from pure states to mixed states as follows.

\begin{definition}[\textbf{Schmidt rank of a density matrix}]
    We define the Schmidt rank of a density matrix $\rhoab$ as follows
    \begin{equation}
    \rk(\rhoab) = \#\supp(\nu(\rhoab)).    
    \end{equation}
    
\end{definition}

In what follows we prove that the Schmidt rank of a density matrix equals its Schmidt number.

\begin{theorem}[\textbf{Equivalence between Schmidt rank and Schmidt number}]
\label{th:equivalence_Schmidtrank_Schmidtnumber}
For any bipartite state $\rhoab$, we have
\begin{equation}
    \rk(\rhoab) = \sn(\rhoab).
\end{equation}

\end{theorem}

\subsection{Example}
\label{sec:example}

We illustrate as an example the Schmidt vectors for the family of states $\rho_\lambda\in\Ha \otimes \Hb$, with $\dim \Ha = \dim \Hb=d$ ,
\begin{equation}
\label{eq:example}
   \rho_\lambda=\left(1-\lambda\right) \frac{\mathbb{1}-\ket{\psi_{\max}}\bra{\psi_{\max}}}{d^2-1}+\lambda \ket{\psi_{\max}}\bra{\psi_{\max}},
\end{equation}
with $\lambda=[0,1]$.
%When $\lambda\leq\frac{1}{d}$, $\rho_\lambda$ is separable. Conversely, if $\lambda>\frac{1}{d}$, it is entangled\cite{Horodecki_rho_lambda}.

Specifically, we numerically obtain the Schmidt vector components $\nu_i(\rho_{\lambda_k})$ for the family of states defined in \eqref{eq:example} for two-qutrits ($d=3$) and for $\lambda_k = 0.05 (k-1)$ with $k=1, \ldots,21$ (see Appendix~\ref{app:numerical}).  
Figure~\ref{fig:figure1}.(left) shows the $i$-th component of the Schmidt vector. 
We observe that $\nu_1$ decreases monotonically up to $1/3$, whereas $\nu_2$ and $\nu_3$ increase monotonically up to $1/3$. 
Additionally, we find that the Schmidt rank is: $\rk(\rho_{\lambda_{k}})=1$ for $\lambda_k \leq 1/3$, $\rk(\rho_{\lambda_{k}})=2$ for $1/3 \leq \lambda_k \leq 2/3$ and $\rk(\rho_{\lambda_{k}})=3$ for $2/3 \leq \lambda_k \leq 1$. From these results, one can conclude that for $\lambda \leq \frac{1}{3}$, $\rho_\lambda$ is separable, while for $\lambda > \frac{1}{3}$, it is entangled. 
This is in agreement with \cite{Horodecki1999}.
Figure~\ref{fig:figure1}.(right) depicts the Lorenz curves associated to each Schmidt vector $\nu\left(\rho_{\lambda_{k}}\right)$. 
We observe that $\| \nu(\rho_{\lambda_k}) - \nu(\rho_{\lambda_{k'}})\|\leq 10^{-3}$  for $k,k'=1,\ldots,7$ and $ \nu(\rho_{\lambda_{k+1}}) \preceq \nu(\rho_{\lambda_{k}})$ for $k=7,\ldots,21$.

\begin{figure}[htb]\centering
	\includegraphics[width=0.9\textwidth]{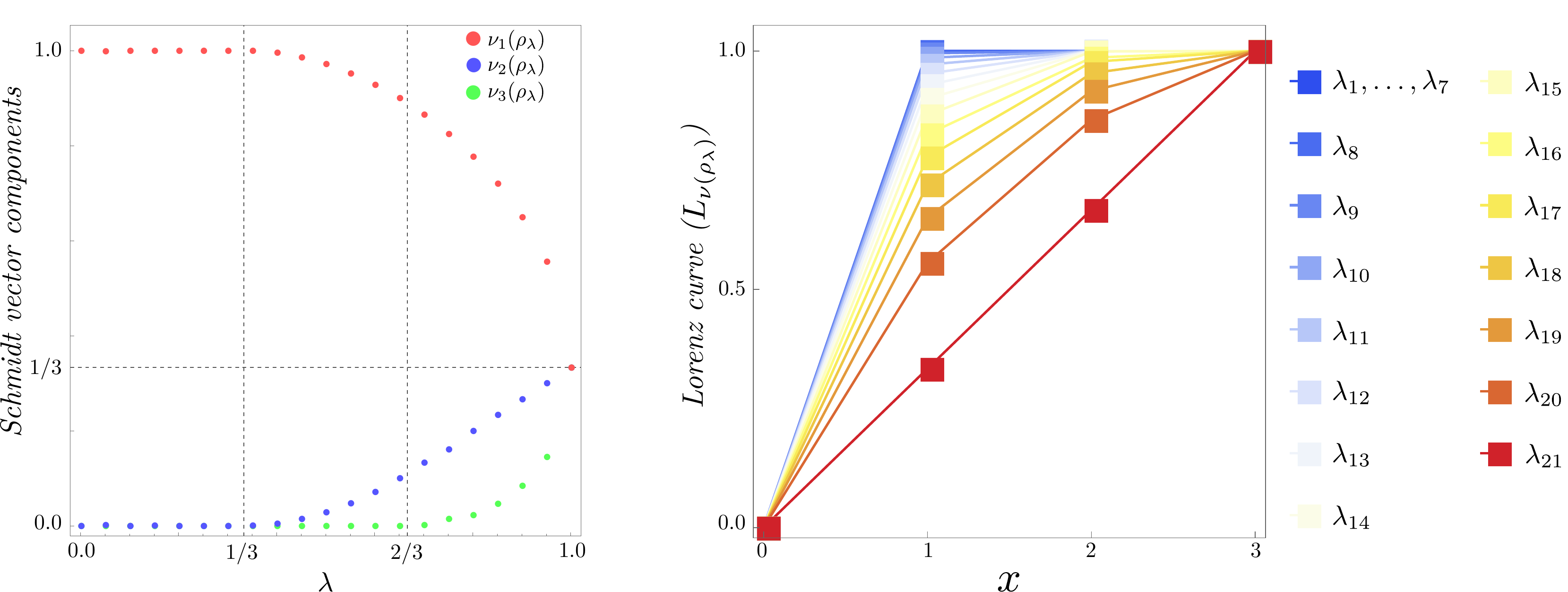}  
	\caption{Schmidt vector components $\nu_i(\rho_{\lambda_k})$ (left) and Lorenz curve $L_{\nu(\rho_{\lambda_k})}$ (right) for the states $\rho_{\lambda_k}$ in~\eqref{eq:example} with $d=3$ and $\lambda_k = 0.05 (k-1)$ where $k=1, \ldots,21$.}
 \label{fig:figure1}
\end{figure}

\subsection{Entanglement monotones from the Schmidt vector}
\label{sec:Schmidt_vector_entanglement_monotone}

We can define a family of entanglement monotones from the Schmidt vector of a quantum state.
\begin{theorem}[\textbf{Schmidt vector entanglement monotone}]
\label{th:Schmidt_vector_entanglement_monotone}
For any $f \in \F$,  $E^{\nu}_{f}: \S(\H^{AB})  \to \R$ given by
	\begin{equation}
	\label{eq:Schmidt_vector_entanglement_monotone}
	E^{\nu}_{f}(\rhoab) =  f(\nu(\rhoab))
	\end{equation}
is an entanglement monotone, that is, it satisfies the conditions~\ref{c1:nonneg}--\ref{c4:normalization}.
   
\end{theorem}

Moreover, if $f$ is also strictly concave, hence strictly Schur-concave, then $E^{\nu}_{f}(\rhoab)=0$ implies that $\rhoab$ is separable.
Additionally, in this case, the maximum of $E^{\nu}_{f}(\rhoab)$ is attained only for maximally entangled states.

The following example shows that the $E_f^{\nu}$ and $E^{\roof}_f$ are distinct families of entanglement monotones.
Consider the family of states $\rho_\lambda$ given in~\eqref{eq:example} and the function $f \in \F$ of the form $f(u) = 1 - u^\downarrow_1+u^\downarrow_d$.
Specifically, for $d=3$ and $\lambda_k = 0.05 (k-1)$ where $k=1, \ldots,21$, we numerically calculate the $E_f^{\nu}(\lambda_k)$ and  $E^{\roof}_f(\lambda_k)$.
The plots in Figure~\ref{fig:figure2} show that these values are numerically distinguishable for  $1/3<\lambda<1$.

\begin{figure}[htb]\centering
	\includegraphics[width=0.45\textwidth]{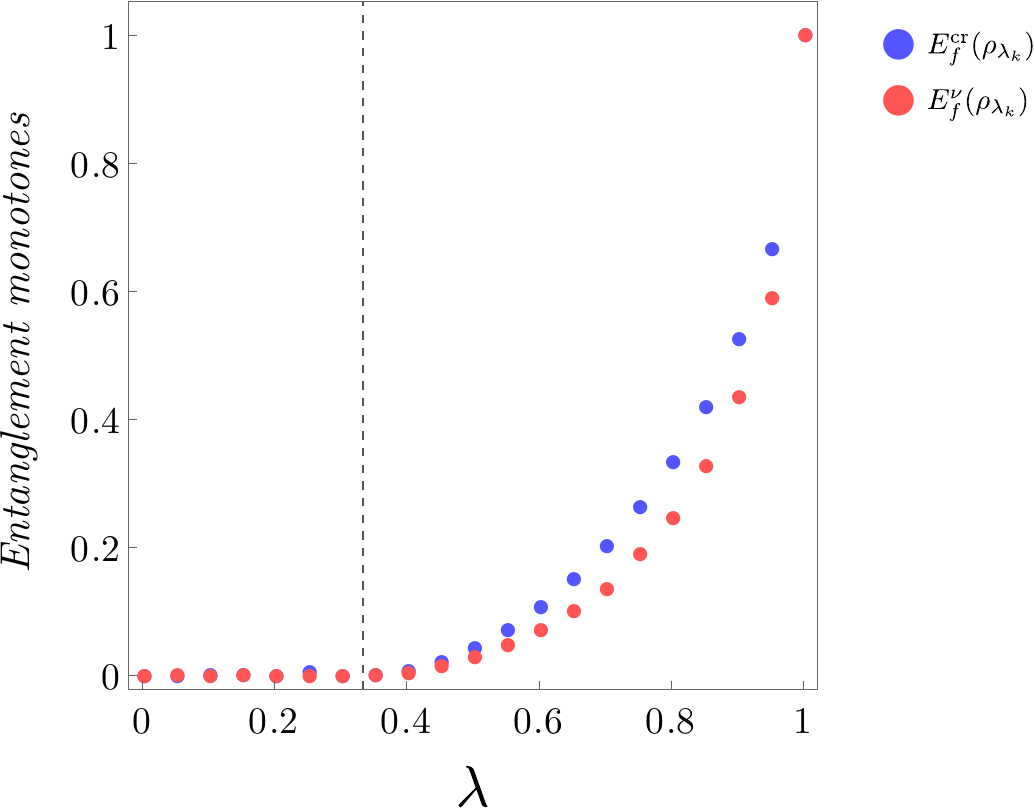}  
	\caption{$E^{\roof}_f$ (blue points) and $E_f^{\nu}$ (red points) for the states $\rho_{\lambda_k}$ in~\eqref{eq:example} with $d=3$ and $\lambda_k = 0.05 (k-1)$ where $k=1, \ldots,21$. The vertical dashed line corresponds to $\lambda=1/3$.}
 \label{fig:figure2}
\end{figure}

\section{Concluding remarks}
\label{sec:conclusions}

In this work, we have addressed the problem of extending the notion of the Schmidt vector from pure states to mixed states, a useful task for characterizing bipartite quantum entanglement. We demonstrated that it is possible to define a Schmidt vector that effectively characterizes entanglement in mixed states using two distinct methods: as a concave roof extension of Schmidt vectors of pure states, or equivalently, from the set of pure states that can be transformed into the mixed state through local operations and classical communication (LOCC).

Our results show that this Schmidt vector fully characterizes separable and maximally entangled states, and it satisfies monotonic and strong monotonic properties under LOCC, providing necessary conditions for conversions between mixed states. 
Additionally, we have extended the definition of the Schmidt rank from pure states to mixed states as the cardinality of the support of the Schmidt vector and have shown that it is equal to the Schmidt number.

Finally, we have introduced a family of entanglement monotones by applying concave and symmetric functions to the Schmidt vector, thereby expanding the tools available for quantifying and characterizing entanglement for mixed quantum states. 

These findings not only enhance the theoretical understanding of entanglement but also could have practical applications in detecting entanglement through entanglement witnesses, as the largest Schmidt coefficient appears in fidelity-based witness operators.  
Further exploration of this connection is warranted, and we leave it open for future research.
Another open problem is extending our results to the multipartite case. 
This extension presents additional complexities and challenges but could further broaden the applicability of the Schmidt vector in quantum information theory.

\begin{acknowledgments}
 GMB acknowledges financial support from project PIBAA 0718 funded by Consejo Nacional de Investigaciones Cient\'ificas y T\'ecnicas CONICET (Argentina).

\end{acknowledgments}

\appendix

\section{Auxiliary lemmas and proofs}
\label{app:aux_proofs}

\subsection{Auxiliary lemmas}

In the first place, we have the auxiliary lemma.

\begin{lemma}
\label{lemma:Sj}
    Given a bipartite state $\rhoab$, for each $1\leq j \leq d$ we consider $S_j = \sup \{s_j(u) : u \in \Ud(\rhoab)\}$, with $s_j(u) = \sum_{i= 1}^j u_i$. 
    There is a pure state decomposition $\{q^{(j)}_i, \ket{\psi^{(j)}_i}\}_{i= 1}^M  \in \D(\rhoab)$ such that
    \begin{equation}
     S_j =  \sum_{i = 1}^{M} q^{(j)}_i s_j(\mu^\downarrow(\ket{\psi^{(j)}_i}\bra{\psi^{(j)}_i})).   
    \end{equation}
    
\end{lemma}

\begin{proof}
For each $u \in \Ud(\rhoab)$ there is a pure state decomposition $\{q_i,\ket{\psi_i}\}_{i=1}^M\in\mathcal{D}(\rhoab)$, such that $u=\sum_{k = 1}^{M}  q_k \mu^\downarrow(\ket{\psi_k}\bra{\psi_k})$.
Then, for each $1 \leq j \leq d$, we have $s_j(u) =\sum_{k = 1}^{M}  q_k s_j(\mu^\downarrow(\ket{\psi_k}\bra{\psi_k}))$. 
Therefore, $S_j = \sup \{\sum_{k = 1}^{M}  q_k s_j(\mu^\downarrow(\ket{\psi_k}\bra{\psi_k})) : {\{q_i,\ket{\psi_i}\}_{i=1}^M\in\mathcal{D}(\rhoab)} \}$.

Since $s_j(\mu^\downarrow (\cdot)) : \P(\Hab) \to [0,1]$ is a continuous function, due to \cite[Prop. 3.5 ]{Uhlmann2010}, there is an optimal pure state decomposition $\{q^{(j)}_i, \ket{\psi^{(j)}_i}\}_{i= 1}^M $ such that
\begin{equation}
S_j =  \sum_{i = 1}^{M} q^{(j)}_i s_j(\mu^\downarrow(\ket{\psi^{(j)}_i}\bra{\psi^{(j)}_i})).    
\end{equation}

\end{proof}

In the second place, we have the auxiliary lemma.
\begin{lemma}
\label{lemma:state_transf_mixed}
Let  $\rhoab$ and $\sigmaab$ be two arbitrary bipartite states. 
Then
\begin{equation}
\begin{split}
    \rhoab\locc\sigma^{AB} \implies & \forall\{q_i,\ket{\psi_i}\}_{i=1}^M\in \D(\rhoab),\exists \{r_l,\ket{\phi_l}\}^{N}_{l=1}\in \D(\sigma^{AB})\  : \\ 
    &\sum_{i=1}^M q_i \mu^\downarrow(\ket{\psi_i}\bra{\psi_i})\preceq \sum^N_{l =1} r_l \mu^\downarrow(\ket{\phi_l}\bra{\phi_l}).
\end{split}
\end{equation}

\end{lemma}

\begin{proof}
    
    Let $\Lambda$ an arbitrary LOCC with Kraus operators $\{K_n\}_{n=1}^N$ such that maps $\rhoab$ into  $\sigma^{AB}$. 
    In other words, $\sigma^{AB}=\Lambda(\rhoab)=\sum_n K_n \rhoab K_n^\dagger$. In addition consider a pure state decomposition of $\rhoab$, $\{q_k, \ket{\psi_k}\}_{k=1}^M$, then
    \begin{equation}
    \sigma^{AB}=\sum_{n=1}^N \sum_{k=1}^M q_k K_n \ket{\psi_k} \bra{\psi_k} K^\dag_n=\sum_{n=1}^{N}\sum_{k=1}^{M} q_k p_{n,k}\ket{\phi_{n,k}}\bra{\phi_{n,k}},    
    \end{equation}
    where $p_{n,k}= \tr(K_n \ket{\psi_k}\bra{\psi_k} K_n^\dagger)\ $ and $\ \ket{\phi_{n,k}}=\frac{K_n\ket{\psi_k}}{\sqrt{p_{n,k}}}$.
    In particular, each $\ket{\psi_k}\bra{\psi_k}$ will be transformed via LOCC to $\sum_{k} p_{n,k} \ket{\phi_{n,k}}\bra{\phi_{n,k}}$. 
    Then, considering Theorem \ref{th:plenio_th}, we obtain
    \begin{equation}
    \mu^\downarrow(\ket{\psi_k}\bra{\psi_k})\preceq \sum_{n=1}^N p_{n,k} \mu^\downarrow(\ket{\phi_{n,k}}\bra{\phi_{n,k}}.    
    \end{equation}
   Now, we can apply the Lemma 32 from  \cite{Bosyk2021}. 
   Thus, there is the following majorization relation between the sum of the vectors:
    \begin{equation}
    \sum_{k=1}^{M} q_k \mu^\downarrow(\ket{\psi_k}\bra{\psi_k})\preceq\sum_{n=1}^{N}\sum_{k=1}^{M} q_k p_{n,k}\mu^\downarrow(\ket{\phi_{n,k}}\bra{\phi_{n,k}}).    
    \end{equation}
    Furthermore $\{q_k p_{n,k},\ket{\phi_{n,k}}\}$ is a pure state decomposition of $\sigma^{AB}$.
    This construction is valid for every pure state decomposition of $\rhoab$.
\end{proof}

\subsection{Proofs of Section~\ref{sec:schimdt_vector}}

Proof of Theorem~\ref{th:equivalence_extensions} (Equivalence of both extensions):

\begin{proof}

\hfill

\begin{enumerate}
    \item  \label{item:ida} Firstly, we show that $\nu_{\D}(\rhoab) \preceq    \nu_{\O}(\rhoab)$.\\
    Given $v\in\Ud(\rhoab)$ there exists a pure state decomposition $\{q_i,\ket{\psi_i}\}_{i=1}^{M}$ of $\rhoab$ such that $v= \sum_{i= 1}^{M} q_i \mu^\downarrow(\ket{\psi_i}\bra{\psi_i})$. Moreover, there is a pure state $\ket{\psi} = \sum_{i=1}^{d} \sqrt{v_i}\ket{a_i b_i} \in \H^{AB} $ such that $\mu^\downarrow(\ket{\psi}\bra{\psi})=v=\sum_n q_i \mu^\downarrow(\ket{\psi_k}\bra{\psi_k})$. 
    By Theorem~\ref{th:plenio_th} we have   $\ket{\psi}\bra{\psi}\locc\rhoab$, then $v\in \Uo(\rhoab)$. Therefore, $\Ud(\rhoab)\subseteq \Uo(\rhoab)$, which implies $\bigvee \Ud(\rhoab) \preceq \bigvee \Uo(\rhoab)$.

    \item \label{item:vuelta} Secondly, we show that $\nu_{\O}(\rhoab) \preceq    \nu_{\D}(\rhoab)$.\\
    Given $u\in\Uo(\rhoab)$ there is a pure state $\ket{\psi} \in \H^{AB}$ such that  $\mu^\downarrow(\ket{\psi}\bra{\psi})=u$ and $\ket{\psi}\bra{\psi}\locc\rhoab$. By Theorem  \ref{th:plenio_th}, there is $u' =\sum_{i=1}^M q_i \mu^\downarrow(\ket{\psi_i}\bra{\psi_i})\in \Ud(\rhoab)$ such that $u \preceq u'$. Therefore, $\forall u\in \Uo$, there is $u'\in \Ud$ such that $u\preceq u'$, which implies  $\bigvee \Uo(\rhoab) \preceq \bigvee \Ud(\rhoab)$.
    
\end{enumerate}
Finally, both results imply $\nu_{\D}(\rhoab) = \nu_{\O}(\rhoab)$.

\end{proof}

Proof of Theorem~\ref{th:Schmidt_properties} (Schmidt vector properties):

\begin{proof}
\hfill
\begin{enumerate}
    \item %Top only attainable for separable states: 

    \begin{itemize}
        \item[$(\Rightarrow)$] 
        Let $\rhoab$ be an arbitrary separable bipartite state. 
        We can express $\rho^{AB}$ as a convex combination of pure product states that act on  $\H^A$ and $\H^B$, respectively, i.e., $\rho^{AB}= \sum_i^M p_i \ket{\psi_i}\bra{\psi_i} \otimes \ket{\phi_i}\bra{\phi_i}$, where $p_i \geq 0$, $\sum_i p_i=1$. 
        Then, ${\left\lbrace  p_i, \ket{ \psi_i \phi_i} \right\rbrace }_{i = 1}^M  \in  \D(\rhoab)$, and $\sum_{i = 1}^{M}  p_i \mu^\downarrow(\ket{\psi_i \phi_i}\bra{\psi_i \phi_i}) = \sum_i p_i(1,0,\ldots,0)=(1,0,\ldots,0) \in \Ud(\rhoab)$. 
        Therefore, $\nu(\rhoab) = \bigvee \Ud(\rhoab)=(1,0, \ldots, 0)$.  
        
        \item[$(\Leftarrow\big)$] Let $\rhoab$ be a bipartite
        state such that $\nu(\rhoab) = (1,0,\ldots,0)$.
        By definition, we have that the first entry of the supremum is given by $L_{\bigvee \Ud(\rhoab)} (1)$, which by construction is equal to 
        \begin{equation}
           \nu_1(\rhoab) = \sup_{\{q_i, \ket{\psi_i}\}_{i= 1}^M  \in \D(\rhoab)}  \left\{ \sum_{i = 1}^{M}  q_i \mu_1^\downarrow(\ket{\psi_i}\bra{\psi_i} )  \right\}.   
        \end{equation}
        Since $\mu_1^\downarrow (\ket{\psi_i}\bra{\psi_i}) = s_1 (\mu^\downarrow (\ket{\psi_i}\bra{\psi_i})$, and due to Lemma \ref{lemma:Sj}, there is an optimal pure state decomposition $\{\tilde{q}_i, \ket{\tilde{\psi}_i}\}_{i= 1}^M \in \D(\rhoab)$ such that
        \begin{equation}
           \nu_1(\rhoab) = \sum_{i = 1}^{M}  \tilde{q}_i \mu_1^\downarrow(\ket{\tilde{\psi}_i}\bra{\tilde{\psi}_i}).   
        \end{equation}
        Then, $1 = \nu_1(\rhoab) =   \sum_{i = 1}^{M}  \tilde{q}_i \mu_1^\downarrow(\ket{\tilde{\psi}_i}\bra{\tilde{\psi}_i})$.Without loss of generality we can assume that $\tilde{q}_i >0$. 
        This is only possible if  $\mu^\downarrow(\ket{\tilde{\psi}_i}\bra{\tilde{\psi}_i}) = (1, 0,\ldots, 0)$ for all  $1\leq i \leq M$. This implies that for for all  $1\leq i \leq M$,  $\ket{\tilde{\psi}_i}\bra{\tilde{\psi}_i}$ is a product state, i.e., $\ket{\tilde{\psi}_i}\bra{\tilde{\psi}_i} = \ket{a_i}\bra{a_i} \otimes \ket{ b_i}\bra{ b_i}$ for some $\ket{a_i }\bra{a_i} \in \S(\H^{A})$ and  $\ket{b_i }\bra{b_i} \in \S(\H^{B})$. 
       Therefore, we have that $\rhoab =\sum_{i = 1}^{M}  \tilde{q_i}     \ket{\tilde{\psi}_i}\bra{\tilde{\psi}_i} = \sum_{i = 1}^{M}  \tilde{q_i}  \ket{a_i}\bra{a_i} \otimes \ket{ b_i}\bra{ b_i}$, so it is separable.

    \end{itemize}

\item %Proof of monotonicity under LOCC:

Let $\rhoab$ be a bipartite state and $\Lambda$ an arbitrary LOCC. We define $\sigmaab = \Lambda(\rhoab)$.
We are going to show that $\nu(\sigmaab)$ is an upper bound of $\Ud(\rho^{AB})$.
Given an arbitrary $u \in \Ud(\rho^{AB})$, we have that there is a pure state decomposition $\{q_i,\ket{\psi_i}\}_{i=1}^M\in \D(\rhoab)$ such that $u = \sum_{i=1}^M q_i \mu^\downarrow(\ket{\psi_i}\bra{\psi_i}) $.
Since $\rhoab \locc \sigmaab$,
due to Lemma~\ref{lemma:state_transf_mixed}, we have that there exists $\{r_l,\ket{\phi_l}\}^{N}_{l =1}\in \D(\sigma^{AB})$ such that $u \preceq \sum^{N}_{l=1} r_l \mu^\downarrow(\ket{\psi_l}\bra{\psi_l})$.
Since $\nu(\sigmaab) = \bigvee \Ud(\sigmaab)$, we have
\begin{equation}
u \preceq \sum^N_{l =1} r_l \mu^\downarrow(\ket{\psi_l}\bra{\psi_l}) \preceq \nu(\sigma^{AB}).    
\end{equation}
This implies that $\nu(\sigma^{AB})$ is an upper bound of $\Ud(\rho^{AB})$. 
Therefore, $\nu(\rhoab) \preceq \nu(\sigmaab)$.

\item %Proof of strong monotonicity under LOCC:

    Let $\rhoab$ be a bipartite state and $\Lambda$ an arbitrary LOCC, given by Kraus opertors $\{K_n\}_{n=1}^{N}$.
    We define $\sigma^{AB}= \Lambda(\rhoab)= \sum_{n=1}^N p_n \sigma_n $, where  $p_n=\tr(K_n \rhoab K_n^\dagger)$ and $\sigma_n=\frac{K_n \rhoab K_n^\dagger}{p_n}$.
    We are going to show that $\sum_{n=1}^N p_n \nu(\sigma_n)$ is an upper bound of $\Ud(\rho^{AB})$.

    Given an arbitrary $u \in \Ud(\rho^{AB})$, we have that there is a pure state decomposition $\{q_i,\ket{\psi_i}\}_{i=1}^M\in \D(\rhoab)$ such that $u = \sum_{i=1}^M q_i \mu^\downarrow(\ket{\psi_i}\bra{\psi_i}) $.
    We can express $\sigma_n=\sum_{i=1}^M \frac{q_i p_{n,i}}{p_n} \ket{\phi_{n,i}}\bra{\phi_{n,i}}$, where $p_{n,i}= \tr(K_n \ket{\psi_i}\bra{\psi_i} K_n^\dagger)$ and $\ket{\phi_{n,i}}=\frac{K_n\ket{\psi_i}}{\sqrt{p_{n,i}}}$.

    On the one hand, we have $\Lambda(\ket{\psi_i} \bra{\psi_i}) = \sum_n K_n\ket{\psi_i} \bra{\psi_i}     
    K^\dag_n = \sum_n  p_{n,i}\ket{\phi_{n,i}}\bra{\phi_{n,i}}$. Then, considering Theorem \ref{th:plenio_th}, we obtain $ \mu^\downarrow(\ket{\psi_i} \bra{\psi_i}) \preceq  \sum_n  p_{n,i} \mu^\downarrow(\ket{\phi_{n,i}}\bra{\phi_{n,i}})$.      
    Then,  from  \cite[Lemma 32]{Bosyk2021}, we infer 
    \begin{equation}
    \label{eq:1era_iii_th10}
    u = \sum_{i=1}^M q_i \mu^\downarrow(\ket{\psi_i}\bra{\psi_i}) \preceq \sum_{n=1} ^N \sum_{i=1}^M q_i p_{n,i} \mu^\downarrow(\ket{\phi_{n,i}}\bra{\phi_{n,i}}).    
    \end{equation}
    On the other hand, since $\sum_{i=1}^M \frac{q_i p_{n,i}}{p_n} \mu^\downarrow(\ket{\phi_{n,i}}\bra{\phi_{n,i}}) \in \Ud(\sigma_n)$,
    we have     
    \begin{equation}
    \sum_{i=1}^M \frac{q_i p_{n,i}}{p_n} \mu^\downarrow(\ket{\phi_{n,i}}\bra{\phi_{n,i}})\preceq \nu(\sigma_n).    
    \end{equation}
    
    Multiplying by $p_n$ and summing over $n$, we obtain
    \begin{equation}
    \label{eq:3era_iii_th10}
    \sum_{n=1} ^N \sum_{i=1}^M q_i p_{n,i} \mu^\downarrow(\ket{\phi_{n,i}}\bra{\phi_{n,i}})\preceq \sum_{n=1}^N p_n \nu(\sigma_n).  
    \end{equation}

    From \eqref{eq:1era_iii_th10} and \eqref{eq:3era_iii_th10} we obtain 
    $u \preceq \sum_{n=1}^N p_n \nu(\sigma_n)$. This implies that  $\sum_{n=1}^N p_n \nu(\sigma_n)$ is an upper bound of $\Ud(\rho^{AB})$.
    Therefore, $\nu(\rho^{AB})\preceq\sum_{n=1}^N p_n \nu(\sigma_n)$.

\item %Proof of bottom only for maximally entangled states:

 \begin{enumerate}
        \item[($\Rightarrow)$] If   $\rhoab$ is of the form $\rhoab = \ket{\psi^{\max}}\bra{\psi^{\max}}$ with $\ket{\psi^{\max}} = \sum^d_{i=1} 1/\sqrt{d} \ket{a_i} \ket{b_i}$, then $\nu(\rhoab)=\mu^\downarrow\left(\ket{\psi^{\max}}\bra{\psi^{\max}}\right) = \left(1/d,\ldots, 1/d\right)$. 
        
        \item[($\Leftarrow)$] Let $\rhoab$ be a bipartite state such that $\nu(\rhoab)=(1/d,\ldots,1/d)$.  
        
        First, we notice that if $\rhoab$ is a pure bipartite state, $\rhoab=\ket{\psi}\bra{\psi}$, it follows that $\ket{\psi}=\ket{\psi^{\max}}$. 
        
        Second, we are going to show that $\rhoab$ has to be a pure state.
        % by appealing to \textit{reductio ad absurdum}.
        Let assume that $\rhoab$ is a mixed state. Let $\{q_i,\ket{\psi_i}\}^M_{i=1}$ be an arbitrary pure state decomposition of $\rhoab$. On the one hand, by definition of supremum, we have that $\sum^M_{i=1} q_i \mu^\downarrow\left(\ket{\psi_i}\bra{\psi_i}\right) \preceq \nu(\rhoab)$. 
        On the other hand,  since $\nu(\rhoab)=(1/d,\ldots,1/d)$ is the bottom of the majorization lattice, we have  $\nu(\rhoab) \preceq \sum^M_{i=1} q_i \mu^\downarrow\left(\ket{\psi_i}\bra{\psi_i}\right)$.
        Then, $v(\rhoab)  = \sum^M_{i=1} q_i \mu^\downarrow\left(\ket{\psi_i}\bra{\psi_i}\right)$. 
        Moreover, since $(1/d,\ldots,1/d)$ is an extreme point of the $d-1$--simplex, then $\mu^\downarrow\left(\ket{\psi_i}\bra{\psi_i}\right)=(1/d,\ldots,1/d)$ for all $i \in \{1, \ldots M\}$.
        This implies that any pure state decomposition of $\rhoab$ is only formed by maximally entangled states.
        
        In particular, let consider the spectral decomposition of $\rhoab$
        \begin{equation}
        \rhoab=\sum_{j=1}^{d^{AB}} \lambda_j \ket{e_j}\bra{e_j}, ~~~ d^{AB} = d^A d^B.    
        \end{equation}
      
        According to the last implication, the eigenvectors $\ket{e_j}$ have to be maximally entangled states. 
        Given that by hypothesis $\rhoab$ is a mixed state, at least two eigenvalues are different from zero. 
        Without loss of generality, we consider $\lambda_1,\lambda_2 > 0$.
        
        We are going to construct another pure state decomposition $\{p_i,\ket{\phi_i}\}^{d^{AB}}_{i=1}$ of $\rhoab$ for which  $\ket{\phi_1}$ is not a maximally entangled state. 
        By using the Schr\"odinger mixture theorem with the unitary matrix $U$
		\begin{equation}
		U = \begin{pmatrix}
		\begin{matrix}
		U_{11} & U_{12} \\
		-U_{12} & U_{11}
		\end{matrix}
		& \rvline & \mathbb{0} \\
		\hline
		\mathbb{0} & \rvline &
		\mathbb{1}_{d^{AB}-2}
		\end{pmatrix},
		\end{equation}
where $U_{11}=\sqrt{\lambda_2/(\lambda_1+\lambda_2)}$ and $U_{12}=\sqrt{\lambda_1/(\lambda_1+\lambda_2)}$, we obtain that $\ket{\phi_1}$ is a superposition of two maximally entangled states, that is,
\begin{equation}
\ket{\phi_1} =\frac{1}{\sqrt{p_1}}(\sqrt{\lambda_1} U_1 \ket{e_1}+\sqrt{\lambda_2}U_2\ket{e_2}).
\end{equation}
Given that the superposition of two maximally entangled states does not result in a maximally entangled state, it follows that $\ket{\phi_1}$ cannot be a maximally entangled state. This presents a contradiction. Consequently, $\rhoab$ cannot be a mixed state; it has to be a pure state.

\end{enumerate}

\end{enumerate}

\end{proof}

\subsection{Proofs of Section~\ref{sec:equivalence_Schmidtrank_Schmidtnumber}}

Proof of Theorem~\ref{th:equivalence_Schmidtrank_Schmidtnumber} (equivalence between Schmidt rank and Schmidt number).

\begin{proof}
Given a bipartite state $\rhoab$, let us define $r = \rk(\rhoab)$ and $n = \sn(\rhoab)$. 

\begin{itemize}

    \item[\textit{(i)}] First, we show that $r\leq n$:

By definition of the Schmidt number, there is a pure state decomposition  $\{p_i,\ket{\psi_i}\}_{i= 1}^{M}$ of $\rhoab$ such that $\rk(\ket{\psi_i}\bra{\psi_i})\leq n$ for all $1\leq i \leq M$. Therefore, $v_n = (1/n, \ldots,1/n,0,\ldots,0)\preceq \mu^\downarrow(\ket{\psi_i}\bra{\psi_i})$ for all $1\leq i \leq M$, which implies $v_n \preceq \sum^M_{i=1} p_i\mu^\downarrow(\ket{\psi_i}\bra{\psi_i})$ (see \cite[Lemma 32]{Bosyk2021}). By definition of the Schmidt vector, we have  $\sum^M_{i=1} p_i\mu^\downarrow(\ket{\psi_i}\bra{\psi_i}) \preceq\nu(\rhoab)$, then $v_n \preceq \nu(\rhoab)$.
This implies that $\nu(\rhoab)$ has a support of at most $n$ components, that is, $r = \#\supp(\nu(\rhoab)) \leq n$.

 \item[\textit{(ii)}] Second, we show that $n\leq r$:
 
If $r= d$, it is trivial, so we suppose $r <d$.
Since $r = \#\supp(\nu(\rhoab))$,  then $\nu_{k}(\rhoab) = 0$, for $r < k \leq d$.

Since $\nu(\rhoab)$ is the supremum of $\Ud(\rhoab)$, its components are given by $\nu_{k}(\rhoab)=  L_{\bigvee \Ud(\rhoab)}(k)-L_{\bigvee \Ud(\rhoab)}(k-1)$.
In particular, for $k = r+1$ we have $L_{\bigvee \Ud(\rhoab)}(r+1) = L_{\bigvee \Ud(\rhoab)}(r)$. Then, by construction of the upper envelope, we conclude that $S_{r} = S_{r+1}$.

Due to Lemma \ref{lemma:Sj} 
there are optimal pure state decompositions $\{q^{(r)}_i, \ket{\psi^{(r)}_i}\}_{i= 1}^{M_{r}}\in \D(\rhoab)$ and  $\{q^{(r+1)}_i, \ket{\psi^{(r+1)}_i}\}_{i= 1}^{M_{r+1}}  \in \D(\rhoab)$ such that:
\begin{align} 
S_{r} &= \sum_{i = 1}^{M_{r}} q^{(r)}_i s_{r}\left(\mu^\downarrow(\ket{\psi^{(r)}_i}\bra{\psi^{(r)}_i})\right),  \\
S_{r+1} &= \sum_{i = 1}^{M_{r+1}} q^{(r+1)}_i s_{r+1}\left(\mu^\downarrow(\ket{\psi^{(r+1)}_i}\bra{\psi^{(r+1)}_i})\right).  
\end{align}
Without loss of generality, we can assume that $q_i^{(r)}>0$ and $q_i^{(r+1)}>0$ for all $1\leq i \leq M$.
             
On the one hand, we have:
\begin{equation}
S_{r} = \sum_{i = 1}^{M_{r}} q^{(r)}_i s_{r}\left(\mu^\downarrow(\ket{\psi^{(r)}_i}\bra{\psi^{(r)}_i})\right) \leq  \sum_{i = 1}^{M_{r}} q^{(r)}_i s_{r+1}\left(\mu^\downarrow(\ket{\psi^{(r)}_i}\bra{\psi^{(r)}_i})\right).
\end{equation}

On the other hand, by definition of the supremum $S_{r+1}$, we have:
\begin{equation}
\sum_{i = 1}^{M_{r}} q^{(r)}_i s_{r+1}\left(\mu^\downarrow(\ket{\psi^{(r)}_i}\bra{\psi^{(r)}_i})\right) \leq S_{r+1}
\end{equation}
Therefore, since $S_r = S_{r+1}$, we have
\begin{equation}
\sum_{i = 1}^{M_{r}} q^{(r)}_i s_{r}\left(\mu^\downarrow(\ket{\psi^{(r)}_i}\bra{\psi^{(r)}_i})\right) =   \sum_{i = 1}^{M_{r}} q^{(r)}_i s_{r+1}\left(\mu^\downarrow(\ket{\psi^{(r)}_i}\bra{\psi^{(r)}_i})\right)
\end{equation}
Equivalently, 
\begin{equation}
\sum_{i = 1}^{M_{r}} q^{(r)}_i \mu_{r+1}^\downarrow(\ket{\psi^{(r)}_i}\bra{\psi^{(r)}_i})  = 0.
\end{equation}
This implies that $\mu_{r+1}^\downarrow(\ket{\psi^{(r)}_i}\bra{\psi^{(r)}_i}) = 0$ for all $1\leq i \leq M_{r}$.
That is, there exist a pure state decomposition $\{q^{(r)}_i, \ket{\psi^{(r)}_i}\}_{i= 1}^{M_{r}}\in \D(\rhoab)$ such that $\rk(\ket{\psi^{(r)}_i}\bra{\psi^{(r)}_i})\leq r$ for all $1\leq i \leq M_{r}$.  
Finally, considering the definition of the Schmidt number, we have that $n\leq r$.

 \item[\textit{(iii)}]  Finally, we conclude that $n= r$, that is, $\rk(\rhoab) = \sn(\rhoab)$.
\end{itemize}

\end{proof}

\subsection{Proofs of Section~\ref{sec:Schmidt_vector_entanglement_monotone}}

Proof of Theorem~\ref{th:Schmidt_vector_entanglement_monotone} (Schmidt vector entanglement monotone):

\begin{proof}
\hfill

\begin{enumerate}
    \item %Top only attainable for separable states: 
    Let $\rhoab$ be an arbitrary bipartite state. We have $\nu(\rhoab) \preceq (1, 0, \ldots, 0)$. Since $f \in \F$, 
    $E_f^{\nu}(\rhoab) = f(\nu(\rhoab))  \geq f(1, 0, \ldots, 0) = 0$.
    Moreover, if  $\rhoab$ is a separable state, from property \ref{c1:sep_sch} of Theorem \ref{th:Schmidt_properties}, $\nu(\rhoab) = (1, 0, \ldots, 0)$, and $E_f^{\nu}(\rhoab) = f(1, 0, \ldots, 0) = 0$.
  
    \item  Let $\rhoab$ be an arbitrary bipartite state and $\Lambda$ an arbitrary LOCC. 
    From property \ref{c2:mono_sch} of Theorem \ref{th:Schmidt_properties}, we have $\nu(\rhoab)\preceq\nu(\Lambda (\rhoab))$.
    Since, $f$ is Schur-concave, $f (\nu(\rhoab)) \geq 
    f(\nu(\Lambda(\rhoab)))$, that is, $E_f^{\nu}(\rhoab) \geq E_f^{\nu}(\Lambda(\rhoab))$.

    \item  Let $\rhoab$ be an arbitrary bipartite state and $\Lambda$ an arbitrary LOCC. 
    We define $\sigma_{AB}=\Lambda(\rho_{AB})=\sum_{n=1}^N p_
n \sigma_n$ with $\sigma_n=\frac{K_n \rhoab K_n^\dagger}{p_n}$ and $p_n=\Tr(K_n \rhoab K_n^\dagger)$, where $\{K_n\}_{n=1}^N$ are the Krauss operators characterizing the LOCC.
    Then, from property~\ref{c3:strong_mono_sch}  of Theorem \ref{th:Schmidt_properties}, we have that
    \begin{equation}
    \nu(\rhoab)\preceq \sum_{n=1}^N p_n \nu(\sigma_n).    
    \end{equation}

    Then,
    \begin{align}
        E_f^{\nu}(\rhoab) &= f(\rhoab) \\
        &\geq f\left(\sum_{n=1}^N p_n \nu(\sigma_n)\right) \\
        &\geq \sum_{n=1}^N p_n f\left(\nu(\sigma_n)\right) =  \sum_{n=1}^N p_n E_f^{\nu}\left(\nu(\sigma_n)\right),
    \end{align}
    where the first inequality follows from the Schur-concavity of $f$, and the second inequality follows from concavity of $f$.

    \item 
    Let $\rhoab_{\max}$ be a maximally entangled state. 
    From property~\ref{c4:maxent_sch} of Theorem \ref{th:Schmidt_properties}, $\nu(\rhoab_{\max})=(1/d, \ldots,1/d)$. 
    Moreover, since $f$ is Schur-concave and $(1/d, \ldots, 1/d) \preceq u$ for all $u \in \Delta_d$, we have $f(u) \leq f(1/d, \ldots, 1/d)$ for all $u \in \Delta_d$. Therefore, for any $\rhoab \in \Hab$,  $f(\nu(\rhoab))\leq f(\nu(\rhoab_{\max}))$, that is,  $E_f^{\nu} (\rhoab) \leq E_f^{\nu} (\rhoab_{\max})$.   
  
\end{enumerate}
\end{proof}

\section{Numerical calculations of the  Schmidt vector}
\label{app:numerical}

In this Appendix, we will illustrate a possible way to obtain the Schmidt vector $\nu(\rhoab) =(\nu_1(\rhoab),\nu_2(\rhoab),\ldots \nu_d(\rhoab))$ numerically. 
This methodology is mainly based on \cite{Alvarez2023}.

The supremum of a set $\U$ can be computed as follows. First, for each $1 \leq j \leq d$, we obtain  $S_j = \sup\{s_j(u) : u \in \U \}$,
where $s_j(u)=\sum^j_{i=1} u_i$. 
Second, we compute the upper envelope of the polygonal curve given by the linear interpolation of the set of points $\{(j,S_{j})\}_{j= 0}^d$, with the convention $S_{0} = 0$.
The result is the Lorenz curve of $\bigvee \U$, $L_{\bigvee \U}$. 
Finally, we have $\bigvee \U = (L_{\bigvee \U)}(1),L_{\bigvee \U}(2)-L_{\bigvee \U}(1),\ldots,L_{\bigvee \U}(d)-L_{\bigvee \U}(d-1))$.

The Schmidt vector of $\rhoab$, $\nu(\rhoab)$, is the supremum of the set $\Ud(\rhoab)$. So, first, we have to compute 
$S_j = \sup\{s_j(u) : u \in \Ud(\rhoab) \}$.

We define $S^M_j = \sup\{s_j(u) : u \in \Ud_M(\rhoab) \}$, with 
\begin{equation}
     \Ud_M(\rhoab) = \left\{ \sum_{i = 1}^{M}  q_i \mu^\downarrow(\ket{\psi_i}\bra{\psi_i}) :      \{q_i, \ket{\psi_i}\}_{i= 1}^M  \in \D(\rhoab) \right\}.
 \end{equation}
It should be noted that $\Ud_M(\rhoab) $ only involves pure state decompositions of at most $M$ states, which implies $S_j^M \leq S_j$. 
For suitable values of $M$, $S^M_j$ is sufficiently close to $S_j$. 
In particular, due to Prop. 2.1 in \cite{Uhlmann2010}, if $M\geq (d^{AB})^2 + 1$ (with $d^{AB} = d^A d^B$), $S_j^M = S_j$.

From the Schrödinger mixture theorem (see, for example, \cite{Nielsen2000}), we have that any ensemble $\{q_i,\ket{\psi_i}\}_{i=1}^M$ is a pure state decomposition of $\rhoab$ if, and only if, there exist a $M\times M$ unitary matrix $V$ such that
\begin{equation}
    \sqrt{q_i}\ket{\psi_i}=\sum_{j=1}^d \sqrt{\lambda_j} V_{ji}\ket{e_j},
    \label{th:schordinger}
\end{equation}
where $\lambda_j$ and $\ket{e_j}$ are the j-th eigenvalue and eigenvector of $\rhoab$, respectively. 

Given $V \in U(M)$ (the set of unitary matrices of dimension $M\times M$), we define $g(V,\rhoab)$ as follows
\begin{equation}
g(V,\rhoab) = \sum_{i=1}^M  \eig^\downarrow_i \left(\Tr_B\left( \sum^d_{j \, j'} \sqrt{\lambda_j \lambda_{j'}}V_{ji}\ket{e_j}\bra{e_{j'}}V_{j'i}^*\right)\right),
\end{equation}
where $\eig^\downarrow(\rho)$ denotes the vector formed by the eigenvalues of $\rho$ sorted in a non-increasing order, and $\eig^\downarrow_i(\rho)$ is the $i$-th entry of $\eig^\downarrow(\rho)$.

The set $\Ud_M(\rhoab)$ can be expressed in terms of the set $U(M)$ as follows
\begin{equation}
\Ud_M(\rhoab)=\left\{g(V,\rhoab)\ :\ V\in U(M) \right\}.    
\end{equation}
Then, $S^M_j  = \sup_{V\in U(K)} \{s_j(g(V,\rhoab)) \} $.
Moreover, any unitary matrix $V$ of dimension $M\times M$ can be expressed as $V = \exp{\imath H}$, with $H$ an Hermitian matrix of dimension $M\times M$, and $H$ can be parameterized as a real linear combination of basis elements of $\mbox{Herm} (M)$, the real vector space of Hermitian matrices. If $\{ A_k\}_{k= 1}^{M^2}$ is a basis of $\mbox{Herm} (M)$, we have $H = \sum_{k= 1}^{M^2} \alpha_k A_k$. We chose the basis $\{ A_k\}_{k= 1}^{M^2}$ as follows:
\begin{equation}
E_{ii}, ~ \frac{E_{ij}+ E_{ij}}{\sqrt{2}}, ~ \frac{\imath (E_{ij}- E_{ij} )}{\sqrt{2}}, ~ 1 \leq i, j \leq M.   
\end{equation}
where $E_{ij}$ is the matrix with only one non-zero entry of value one in the $i$-th row and $j$-th column. 
Then, 
$S^M_j  = \sup_{\vec{\alpha}\in \mathbb{R}^{M^2}} \{s_j(g(V(\vec{\alpha}),\rhoab)) \} $, with $V(\vec{\alpha}) = \exp{\imath \sum_{k=1}^{M^2} \alpha_k A_k}$
For $j= 1, \ldots, d-1$, we obtained numerically the supremum $S^M_j$, which approximates $S_j$.

Finally, the numerical approximation of the Schmidt vector is obtained from the upper envelope of the linear interpolation of the points $\{j,S^M_j\}_{j=0}^d$, with $S^M_0=0$.

%%%%%%%%%%%%%%

\end{document}